%% file: sn-article.tex
\documentclass[pdflatex,sn-aps]{sn-jnl}

\usepackage{graphicx}%
\usepackage{multirow}%
\usepackage{amsmath,amssymb,amsfonts}%
\usepackage{amsthm}%
\usepackage{mathrsfs}%
\usepackage[title]{appendix}%
\usepackage{xcolor}%
\usepackage{textcomp}%
\usepackage{manyfoot}%
\usepackage{booktabs}%
\usepackage{algorithm}%
\usepackage{algorithmicx}%
\usepackage{algpseudocode}%
\usepackage{listings}%

\usepackage[whole]{bxcjkjatype}
\usepackage{dcolumn}
\usepackage{bm}
\usepackage{amsthm} 
\usepackage{comment}

\newcommand{\red}{\color{red}}

\theoremstyle{plain}
\newtheorem{thm}{Theorem}
\newtheorem*{thm*}{Theorem}

\theoremstyle{definition}

\begin{document}

\title[Article Title]{
One-shot prediction of noise-induced bifurcations with reservoir computing
}


\author*[1]{\fnm{Nozomi} \sur{Akashi}}\email{akashi.nozomi.2a@kyoto-u.ac.jp}

\author[2]{\fnm{Takayuki} \sur{Watanabe}}

\author[3]{\fnm{Masato} \sur{Hara}}

\author[4]{\fnm{Takao} \sur{Namiki}}

\author[3]{\fnm{Hiroshi} \sur{Kokubu}}

\author[5]{\fnm{Ichiro} \sur{Tsuda}}

\author[6]{\fnm{Kohei} \sur{Nakajima}}

\affil*[1]{\orgdiv{Graduate School of Informatics}, \orgname{Kyoto University}, \orgaddress{\street{Yoshida-honmachi}, \city{Sakyo-ku}, \postcode{606-8501}, \state{Kyoto}, \country{Japan}}}
\affil[2]{\orgdiv{College of Science and Engineering}, \orgname{Chubu University}, \orgaddress{\street{1200 Matsumoto-cho}, \city{Kasugai}, \postcode{487-8501}, \state{Aichi}, \country{Japan}}}
\affil[3]{\orgdiv{Department of Mathematics}, \orgname{Kyoto University}, \orgaddress{\street{Kitashirakawa Oiwake-cho}, \city{Sakyo-ku}, \postcode{606-8502}, \state{Kyoto}, \country{Japan}}}
\affil[4]{\orgdiv{Faculty of Science}, \orgname{Hokkaido University}, \orgaddress{\street{Kita-10 Nishi-8, Kita-ku}, \city{Sapporo}, \postcode{060-0810}, \state{Hokkaido}, \country{Japan}}}
\affil[5]{\orgdiv{AIT Center}, \orgname{Sapporo City University}, \orgaddress{\street{1 Geijutsu-no-Mori, Minami-ku}, \city{Sapporo}, \postcode{005-0864}, \state{Hokkaido}, \country{Japan}}}
\affil[6]{\orgdiv{Graduate School of Information Science and Technology}, \orgname{The University of Tokyo}, \orgaddress{\street{7-3-1 Hongo}, \city{Bunkyo-ku}, \postcode{113-8654}, \state{Tokyo}, \country{Japan}}}


\input{section/0_abstract}




\maketitle

\input{section/1_introduction}
\input{section/2_methods}
\input{section/3_experimentalResults}

\input{section/6_spin}

\input{section/7_conclusion}

\backmatter

\bmhead{Acknowledgements}
Nozomi Akashi, Takayuki Watanabe, and Ichiro Tsuda were supported by JSPS KAKENHI Grant No. 25K00011.
Mr. Sora Todaka assisted with parameter exploration for numerical experiments.

\begin{appendices}

\input{section/A2_benchmarks}
\input{section/A2_experimentalSetting}
\input{section/A1_proof}
\input{section/A3_spinParameter}

\bmhead{Supplementary information}

\input{section/S1_spin_PARC}




\end{appendices}


\bibliography{bib/reservoir, bib/physicalReservoirComputing, bib/neuralnetwork, bib/spintronics, bib/nonlinearDynamics}

\end{document}

%% file: section/0_abstract.tex
\abstract{
Dynamical systems can exhibit complex responses when noise is injected. 
In particular, dynamics can be qualitatively altered by dynamic noise, a phenomenon known as noise-induced bifurcation.
Predicting noise-induced bifurcations is a critical challenge in nonlinear physics.
Recently, it has been reported that reservoir computing, a machine learning framework, can reconstruct the unseen global structure of a dynamical system, including bifurcations, from limited time series data.
However, learning global structures in random dynamical systems has not yet been systematically addressed.
In this study, we report that a simple reservoir computing framework can predict the noise-induced bifurcation structure from the time series at a single noise condition.
We demonstrate dynamic noise cancellation and the reconstruction of entire noise-induced bifurcation structures, including noise-induced chaos and noise-induced order, in representative dynamical systems.
Additionally, we provide a theoretical explanation for noise cancellation and demonstrate noise cancellation of a neuromorphic spintronics device.
Our results provide significant insights into understanding and harnessing real-world noisy complex dynamics.
}

%% file: section/1_introduction.tex
\section{\label{sec:intro}Introduction}
Physics is a discipline that seeks to uncover laws on the basis of observation.
The behavior of systems exposed to inevitable dynamical noise can sometimes differ markedly from that of the underlying intrinsic dynamics. 
Such effects are known as noise-induced phenomena, and include noise-induced bifurcation, in which the stationary state of a system changes~\cite{crutchfield1982fluctuations, matsumoto1983noise}, as well as noise-induced synchronization~\cite{toral2001analytical}, noise-induced transitions~\cite{horsthemke1984noise}, and stochastic resonance~\cite{gammaitoni1998stochastic}.
These phenomena have been demonstrated in simple mathematical models and have also been widely observed in real-world systems~\cite{matsumoto1983noise, horsthemke1984noise, gammaitoni1998stochastic, akashi2020input-driven}.
Accurately predicting the onset of such noise-induced phenomena, as well as suppressing them, has long been a major challenge in nonlinear physics. 
In real-world signals, such as electroencephalogram data, clean noise-free reference data are often unavailable. 
When noise-induced bifurcations occur, the properties of the observed signals may differ substantially from those of the original signals; for example, the dynamics may shift from chaos to order.
Therefore, noise cancellation and prediction of noise-induced bifurcations from a single noise condition can be regarded as important technological milestones for both the understanding of complex real-world dynamics and those engineering applications.

In recent years, data-driven approaches have shown considerable promise for learning dynamical systems. 
Among them, reservoir computing~\cite{jaeger2001echo, maass2002real, nakajima2021reservoir} has proven particularly effective for predicting complex trajectories, including chaotic dynamics~\cite{pathak2018model, pathak2018hybrid, vlachas2020backpropagation}. 
A key feature of reservoir computing is its simple training scheme, in which only the readout layer is optimized.
Beyond trajectory prediction, reservoir computing has been shown both experimentally and theoretically to capture structural properties of the dynamical system underlying the data.
Experiments have demonstrated that trained reservoirs can reproduce invariant features such as invariant distribution and Lyapunov spectra of attractors~\cite{pathak2017using, lu2018attractor}, while theory has related the learned reservoir dynamics to the original system through conjugacy-like constructions~\cite{hara2024learning}. 
Moreover, reservoir computing can, in some cases, infer aspects of the global structure of a dynamical system that are not explicitly represented in the training data.
In particular, reservoirs have been reported to predict bifurcation points and post-bifurcation dynamics from data acquired only in the pre-bifurcation regime~\cite{itoh2017reconstructing, itoh2020reconstructing, kim2021teaching, hara2024learning, kong2023reservoir, tokuda_prediction_2024, tadokoro2024trans, shen2026predicting}. 

Reservoir computing has also been used to learn noisy time series~\cite{semenova2019fundamental, estebanes2019constructive, itoh2020reconstructing, wikner2024stabilizing, luo2024reconstructing}. 
Recent studies have shown that it can be used for denoising~\cite{nathe2023reservoir, Sedehi2025denoising} and noise separation~\cite{choi2025signal}.
However, most of these approaches face practical limitations, as they require clean, noise-free data for training or assume observational noise that does not alter the intrinsic properties of the underlying dynamical system. 
Another study~\cite{lin_learning_2024} shows that a reservoir with Sparse Identification of Nonlinear Dynamics (SINDy)~\cite{brunton2016discovering} can reconstruct noise-induced transitions, which occur when noise drives trajectories across basin boundaries separating attractors that already exist in the noise-free system.
Noise-induced bifurcations, by contrast, require the inference of more global phase space structure, because they involve the emergence of attractors that are absent without noise.
Taken together, these studies suggest that the concise time-series learning framework of reservoir computing enables reconstruction of both denoised dynamics and bifurcation structure from data corrupted by dynamic noise, even in the presence of noise-induced bifurcations.

Here, we show that, by leveraging the ability of reservoir computing to reconstruct the underlying structure of time series data, noise-induced bifurcations can be predicted in a one-shot manner from training data acquired under a single noise condition. 
In our approach, time series corrupted by dynamic noise are used for training, and the system's behavior is then predicted as the noise intensity is varied. 
Using representative dynamical systems known to exhibit noise-induced bifurcations, we numerically demonstrate that this framework enables both reconstruction of the underlying noise-free dynamics and prediction of the associated bifurcation structure.
We further clarify theoretically the conditions under which such dynamic noise cancellation succeeds.
Finally, as an application to a real physical system exhibiting noise-induced phenomena, we demonstrate dynamic noise cancellation for noise-induced bifurcations in a spin-torque oscillator model~\cite{akashi2020input-driven}, a promising device for neuromorphic computing. 
These results open a route to predicting, understanding, and harnessing complex dynamical responses induced by real-world noise.

%% file: section/2_methods.tex
\section{Proposed Framework}

We consider random dynamical systems driven by noise of the form
\begin{eqnarray}
    \label{eq:generalNoise}
    {\bm y}_{\varepsilon}(t+1) = {\bm f}({\bm y}_{\varepsilon}(t), \varepsilon {\bm v}(t+1)),
\end{eqnarray}
where ${\bm y}_{\varepsilon}(t) \in \mathbb{R}^n$ denotes the state of the system under noise intensity $\varepsilon > 0$, and ${\bm v}(t) \in \mathbb{R}^n$ denotes a noise sequence drawn from a given probability distribution. 
We also focus on the special case of additive noise, given by
\begin{eqnarray}
    \label{eq:additiveNoise}
    {\bm y}_{\varepsilon}(t+1) = {\bm f}({\bm y}_{\varepsilon}(t)) + \varepsilon {\bm v}(t+1).
\end{eqnarray}

\begin{figure*}[htbp]
\begin{center}
\includegraphics[width=\linewidth]{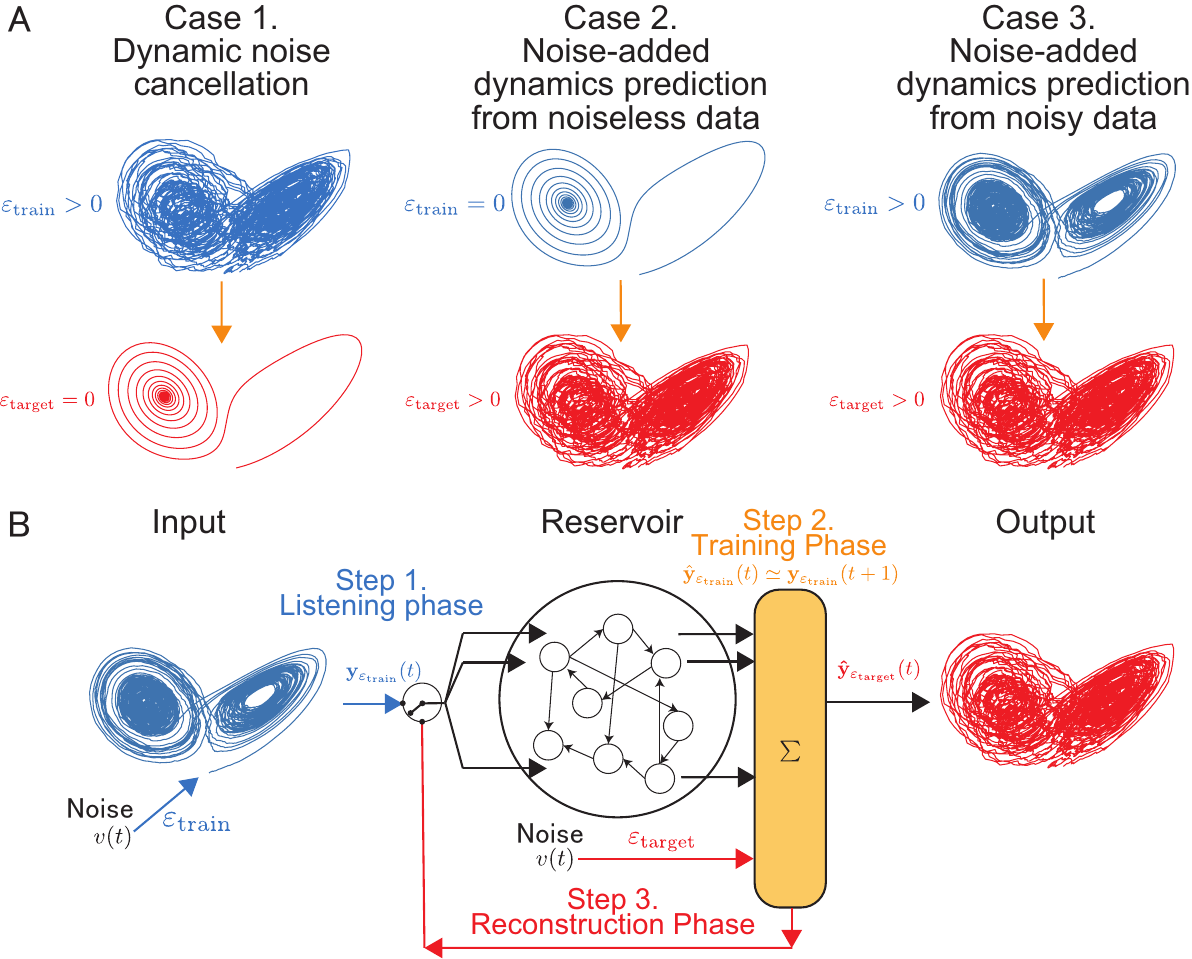}
\caption{\label{fig:schematics}
Schematic illustration of the proposed framework.
(A) Classification of its use cases.
(B) Overview of the framework and its procedure.
}
\end{center}
\end{figure*}

We next describe the scope of the proposed method. In Eqs.~(\ref{eq:generalNoise}) and (\ref{eq:additiveNoise}), we denote by $\varepsilon_{\rm train}$ the noise intensity in the observed data used for training, and by $\varepsilon_{\rm target}$ the noise intensity in the dynamics to be reconstructed.
Depending on the combination of $\varepsilon_{\rm train}$ and $\varepsilon_{\rm target}$, the use cases of the method can be classified into the following four categories, as shown in Fig.~\ref{fig:schematics}A.
\begin{enumerate}
    \setcounter{enumi}{-1}
    \item \textbf{Case 0. Attractor reconstruction:} $\varepsilon_{\rm train}=0$, $\varepsilon_{\rm target}=0$.  
    This corresponds to standard attractor reconstruction in reservoir computing~\cite{jaeger2001echo, verstraeten2007experimental, pathak2017using, lu2018attractor, hara2024learning}. The aim is to reconstruct the stationary dynamics represented in the training data.

    \item \textbf{Case 1. Dynamic noise cancellation:} $\varepsilon_{\rm train}>0$, $\varepsilon_{\rm target}=0$.  
    The aim is to reconstruct the noiseless dynamics from training time series corrupted by dynamic noise.

    \item \textbf{Case 2. Noise-added dynamics Prediction from noiseless data:} $\varepsilon_{\rm train}=0$, $\varepsilon_{\rm target}>0$.  
    The aim is to predict the dynamics at a non-zero noise intensity $\varepsilon_{\rm target}$ from noiseless training data.

    \item \textbf{Case 3. Noise-added dynamics Prediction from noisy data:} $\varepsilon_{\rm train}>0$, $\varepsilon_{\rm target}>0$.  
    The aim is to predict the dynamics at a target noise intensity $\varepsilon_{\rm target}$ from training data obtained at a different non-zero noise intensity $\varepsilon_{\rm train}$.
\end{enumerate}

In Case 1, we exploit the fact that the reservoir can function as a dynamic noise-removal filter even without any explicitly designed denoising procedure.
In Case 2, if the reservoir successfully generalizes the underlying dynamics around the training data, it can reconstruct the noise response by evolving consistently with the target dynamical system, even when noise drives the state away from the attractor present in the noiseless system.
Case 3 can be regarded as an integration of Cases 1 and 2: denoising is effectively realized in the training phase, while extrapolation to a different noise intensity is achieved in the reconstruction phase.
In Cases 2 and 3, the entire noise-induced bifurcation structure is predicted by reconstructing the family of dynamical behaviors associated with different noise intensities.

We use echo state networks (ESNs)~\cite{jaeger2001echo, jaeger2004harnessing} as a reservoir computing model. 
The time evolution of the ESN is given by
\begin{eqnarray}
{\bm x}(t) &=& {\bm g}\left(\rho W{\bm x}(t-1) + \sigma W^{\rm in}{\bm u}(t) + {\bm b}\right), 
\label{eq:ESN}
\end{eqnarray}
where reservoir states at time step $t$ are represented as ${\bm x}(t) \in {\mathbb R}^{d}$.
The input is denoted by ${\bm u}(t) \in {\mathbb R}^{d_{\rm in}}$.
The activation function ${\bm g}\colon \mathbb{R}^d \to \mathbb{R}^d$ is taken to act element-wise. 
The entries of the input-weight matrix $W^{\rm in} \in \mathbb{R}^{d\times d_{\rm in}}$, the recurrent weight matrix $W \in \mathbb{R}^{d\times d}$, and the bias vector ${\bm b} \in \mathbb{R}^d$ are sampled independently from the uniform distribution on $[-1,1]$.
The recurrent weight matrix $W$ is then rescaled so that its spectral radius is unity.
The spectral radius and input scaling are controlled by the hyperparameters $\rho$ and $\sigma$, respectively. 
These reservoir configurations $W_{\rm in}, W$, ${\bm b}$, $\rho$, and $\sigma$ are not optimized through the training procedure in reservoir computing. 
Detailed settings are provided in the Appendix~\ref{sec:setting}.

In this study, attractor reconstruction using reservoir computing consists of three phases: (1) listening phase, (2) training phase, and (3) reconstruction phase, shown in Fig.~\ref{fig:schematics}B.
In the listening phase, a time series ${\bm y}_{\varepsilon_{\rm train}}(t)$ generated by the target dynamical system at a given noise intensity $\varepsilon_{\rm train}$ is fed into the ESN as the input ${\bm u}(t)$. 
The resulting reservoir state driven by the noisy signal ${\bm y}_{\varepsilon_{\rm train}}(t)$ is denoted by ${\bm x}_{\varepsilon_{\rm train}}(t)$.
The total length of the input time series is $T_{\rm wash}+T_{\rm train}$, where $T_{\rm wash}$ is the washout period used to remove the influence of the initial reservoir state, and $T_{\rm train}$ is the training duration used.

In the training phase, we determine the readout matrix $W^{\rm out}\in \mathbb{R}^{(d+1)\times d_{\rm out}}$ so that the ESN output
\begin{eqnarray}
\label{eq:trainingOutput}
\hat{{\bm y}}_{\varepsilon_{\rm train}}(t) &=& [{\bm x}^{\top}_{\varepsilon_{\rm train}}(t); 1]W^{\rm out}
\end{eqnarray}
predicts the next step of the training signal ${\bm y}_{\varepsilon_{\rm train}}(t+1)$. Specifically, $W^{\rm out} $ is obtained by ridge regression:
\begin{eqnarray}
\label{eq:ridgeregression}
W^{\rm out}
&=&
\underset{W^{\rm out}}{\rm argmin}
\sum_{t=T_{\rm wash}}^{T_{\rm wash}+T_{\rm train}-1}
\left\|
{\bm y}_{\varepsilon_{\rm train}}(t+1)
-
\hat{{\bm y}}_{\varepsilon_{\rm train}}(t)
\right\|^2 \nonumber +
\beta \|W^{\rm out}\|_{\rm F}^2, 
\end{eqnarray}
where $\beta > 0$ is the regularization parameter, and $\|\cdot\|_{\rm F}$ denotes the Frobenius norm.

In the reconstruction phase, we construct an autonomous dynamical system by adopting a feedback configuration in which the reservoir output is fed back as the input at the next time step.
The aim is for the projected output $\hat{{\bm y}}(t)$ of this autonomous reservoir to reconstruct the target dynamical system. 
To reconstruct the dynamics at a target noise intensity $\varepsilon_{\rm target}$, noise of the same intensity is added to the output.
These procedures are summarized by
\begin{eqnarray}
\label{eq:closedloop}
\hat{{\bm y}}_{\varepsilon_{\rm target}}(t) &=& [{\bm x}^{\top}(t); 1]W^{\rm out} + \varepsilon_{\rm target}{\bm v}(t), \\
{\bm u}(t+1) &=& \hat{{\bm y}}_{\varepsilon_{\rm target}}^{\top}(t).
\end{eqnarray}
Through this feedback mechanism, the reservoir output $\hat{{\bm y}}_{\varepsilon_{\rm target}}(t)$ is expected to reconstruct the behavior of the target dynamical system under the unseen noise intensity $\varepsilon_{\rm target}$.

%% file: section/3_experimentalResults.tex
\section{Results}
In this section, we demonstrate that the proposed framework can cancel dynamical noise and predict noise-added dynamics in typical dynamical systems.
We consider three types of noise-induced bifurcations: noise-induced chaos, noise-induced order, and two-step noise-induced bifurcation, in dynamical systems with dynamic noise, which serve as benchmarks for evaluating the performance of the proposed framework: the random logistic map~\cite{matsumoto1983noise, sano2020reduction}, the random Belousov–Zhabotinsky (BZ) map~\cite{matsumoto1983noise}, and the stochastic Lorenz system~\cite{toral2001analytical}. 
Details of the benchmark systems and experimental settings are provided in Appendices~\ref{sec:benchmark} and \ref {sec:setting}, respectively.

\subsection{Dynamic noise cancellation}

\begin{figure*}[htbp]
\begin{center}
\includegraphics[width=\linewidth]{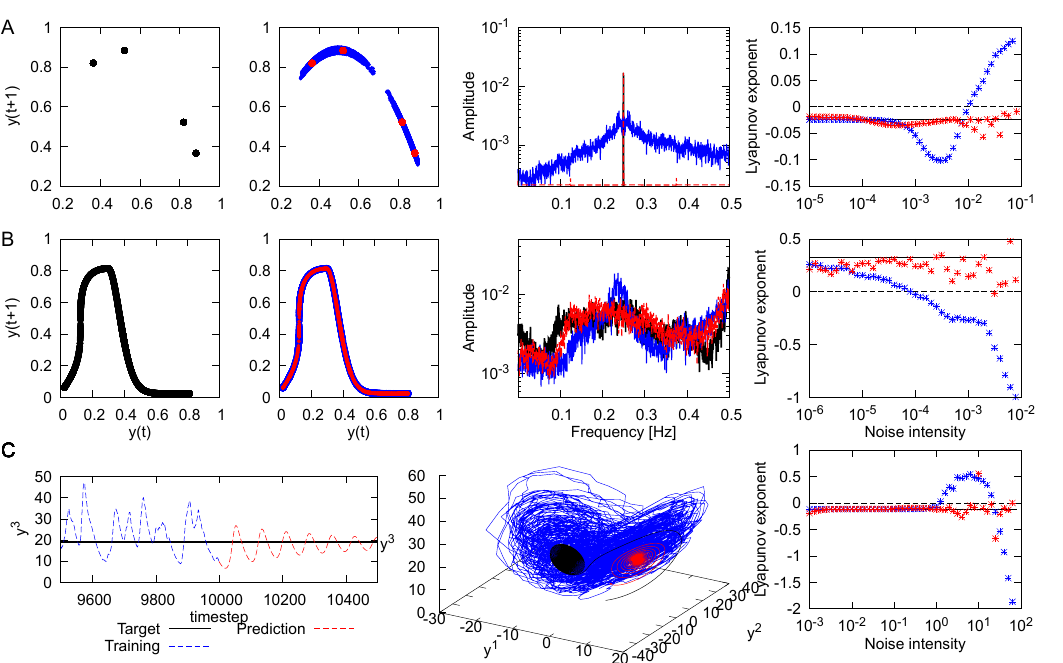}
\caption{
\label{fig:noiseCancellaration}
Results of dynamic noise cancellation.
Black, blue, and red points and lines indicate those for the noiseless target, noisy training, and reconstructed signals, respectively.
A, B. Results for the logistic map and the BZ map.
In each row, the panels from left to right present the return map of the noiseless target, the return map of the noisy training and reconstruction data, the Fourier spectra, and the Lyapunov exponent as a function of noise intensity.
C. Results for the Lorenz system.
The panels from left to right present the time series, the reconstructed trajectory in phase space, and the Lyapunov exponent as a function of noise intensity.
In the time-series panel, the ESN is switched from open-loop to closed-loop operation at time step 10,000.
We set the parameters $(\rho, \sigma, \varepsilon_{\rm train})=(0.05, 1.00, 0.0158), (0.00, 15.85, 0.00158)$, and $(0.30, 1.00, 31.62)$ in A, B, and C, respectively.
}

\end{center}
\end{figure*}

This section presents an analysis of dynamic noise cancellation in the three benchmark systems. 
By comparing the geometry of the reconstructed attractors and the Fourier spectra produced by the reservoir, under representative ESN parameters, with those of the target systems, we assess the spatiotemporal reconstruction capability of the method. 
We then compute the Lyapunov exponents of the reservoir in the reconstruction phase and compare them with those of the noise-free target systems to evaluate how accurately the noise-induced bifurcation structure and bifurcation points are reconstructed. 
Finally, by varying the noise intensity used during training, we examine the robustness of the reconstruction using the maximum Lyapunov exponent, which serves as an indicator of bifurcation, as the evaluation metric.

Figure~\ref{fig:noiseCancellaration}A shows the logistic-map result. Even when the ESN is trained on noisy trajectories in the noise-induced chaotic regime, it reconstructs the noiseless period-4 orbit accurately in both geometry and temporal evolution. 
The Lyapunov exponent is also reproduced well before the bifurcation, and remains close to the target value even in the post-bifurcation regime.
Figure~\ref{fig:noiseCancellaration}B shows the BZ-map result. 
Training is performed on trajectories in the noise-induced ordered regime, yet the ESN reconstructs the original chaotic dynamics.
The 0.25-Hz peak in the power spectrum associated with the noise-induced order in the learned dynamics disappears in the noise-canceled dynamics, and the Lyapunov exponent is correctly recovered as positive over nearly the entire range considered.
Figure~\ref{fig:noiseCancellaration}C shows the Lorenz system result. The ESN is trained on trajectories regularized by the two-step bifurcation and reconstructs the target orbit converging to a fixed point. While reconstruction fails at some specific noise intensities, the overall two-step bifurcation structure of the Lyapunov exponent is recovered.

\input{section/4_theoreticalAnalysis}

\subsection{Noise-added dynamics prediction from noiseless data}
\begin{figure*}[htbp]
\begin{center}
\includegraphics[width=\linewidth]{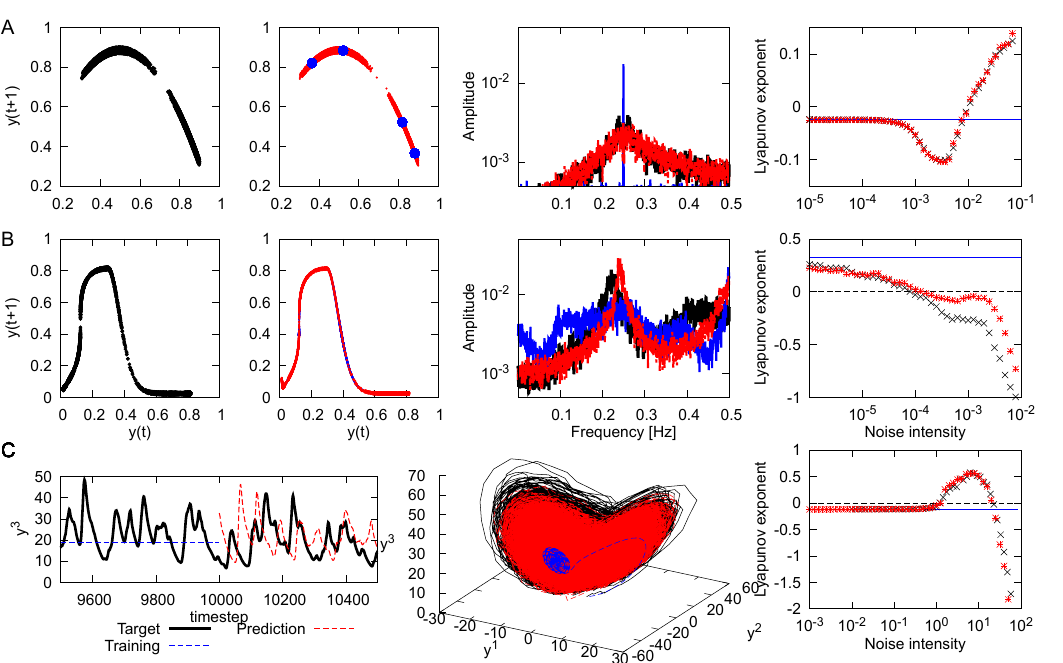}
\caption{
\label{fig:nibReconstruction}
Results of noise-added dynamics prediction from noiseless data.
Black, blue, and red points and lines indicate the noisy target, noiseless training, and predicted signals, respectively.
A, B. Results for the logistic map and the BZ map.
In each row, the panels from left to right present the return map of the noiseless target, the return map of the noisy training and reconstruction data, the Fourier spectra, and the Lyapunov exponent as a function of noise intensity.
C. Results for the Lorenz system.
The panels from left to right present the time series, the reconstructed trajectory in phase space, and the Lyapunov exponent as a function of noise intensity.
In the time-series panel, the ESN is switched from open-loop to closed-loop operation at time step 10,000.
We set the parameters $(\rho, \sigma, \varepsilon_{\rm target})=(0.05, 1.00, 0.0158), (0.00, 15.85, 0.00759)$, and $(0.00, 1.00, 31.62)$ in A, B, and C, respectively.
}    
\end{center}
\end{figure*}

We next train the ESN on noiseless trajectories and test whether noise injection during prediction reproduces the corresponding noise-induced bifurcations.

For the logistic map (Fig.~\ref{fig:nibReconstruction}A), the ESN learns the period-4 orbit accurately, and injected noise induces chaos in the reconstructed dynamics above the correct noise scale. Both the attractor geometry and the Lyapunov exponent are reproduced well.
For the BZ map (Fig.~\ref{fig:nibReconstruction}B), the reconstructed attractor remains visually similar to that of the noiseless dynamics, but the power spectrum and Lyapunov exponent indicate that the ESN captures the noise-induced order qualitatively, including the bifurcation point.
For the stochastic Lorenz system (Fig.~\ref{fig:nibReconstruction}C), the ESN is trained on transient trajectories converging to the noiseless fixed point. When noise is injected during reconstruction, the ESN reproduces both attractors associated with the two-step bifurcation and captures the corresponding change in Lyapunov exponent with high accuracy.

\subsection{Noise-added dynamics prediction from noisy data}
\begin{figure*}[htbp]
\begin{center}
\includegraphics[width=\linewidth]{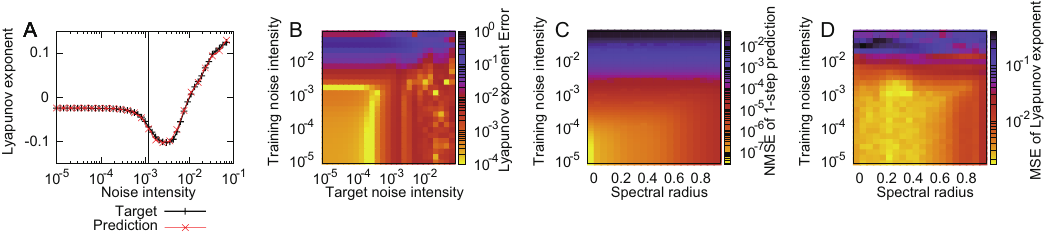}
\caption{
\label{fig:advanced}
Prediction of the noise-induced bifurcation in the logistic map from noisy signals. 
The results are the average values across five different random realizations of ESNs.
A, Reconstruction of the Lyapunov exponent from training data with noise intensity $\varepsilon_{\rm train}=0.00120$. 
B, Error map for the reconstructed Lyapunov exponent as a function of the training and target noise intensities, $\varepsilon_{\rm train}$ and $\varepsilon_{\rm target}$. 
C, Normalized mean squared error map of the one-step prediction in training phase as a function of the spectral radius $\rho$ and the training noise intensity $\varepsilon_{\rm train}$.
D, Mean squared error map of the reconstructed Lyapunov exponent, averaged over the target noise intensity $\varepsilon_{\rm target}$, as a function of the spectral radius $\rho$ and the training noise intensity $\varepsilon_{\rm train}$.
We set the reservoir parameters $(\rho, \sigma)=(0.3, 1.00)$ in A and B.
}    
\end{center}
\end{figure*}

We next examine how prediction performance depends on the training data and reservoir configuration, using the random logistic map as a test case.
This is achieved by combining dynamical noise cancellation during training with the prediction of the noise-added dynamics during reconstruction.

Figure~4A shows that the ESN reconstructs the Lyapunov exponent across the bifurcation structure with high accuracy even when trained on noisy data. 
Figure~4B shows that accurate reconstruction is obtained as long as the training noise remains below the onset of the bifurcation.
Interestingly, the smallest error is not achieved in the limit of vanishing training noise, but around $\varepsilon_{\rm train}\approx 10^{-3}$, suggesting that moderate noise can improve reconstruction by driving trajectories through a wider region of phase space.

The predicted noise-free and noisy dynamics obtained by the proposed method differ from the dynamics at the noise intensity presented during training. 
This raises the question of how one can justify that the predicted dynamics indeed correspond to those altered by the underlying noise-induced bifurcation.
Figure~\ref{fig:advanced}C and D show the performance in phase 2 (training) and phase 3 (prediction), respectively, as functions of the spectral radius, which governs the memory capacity of the reservoir, and the training noise intensity. 
Both panels exhibit the same overall trend: the performance is better in the regime of small spectral radius and low training noise intensity.
This suggests that a small error in the training phase is expected to lead to good predictive performance for dynamics modified by noise.
The smallest error of Lyapunov exponent is obtained around $\rho=0.3$, rather than at $\rho=0$, although the target system is Markovian.
This suggests that reservoir memory provides redundancy for one-step prediction and improves robustness in the closed-loop reconstruction.
It should be noted, however, that this tendency cannot be inferred from the training results in panel C alone.

%% file: section/4_theoreticalAnalysis.tex
\subsubsection{Theoretical analysis \label{sec:theory}}
Here, we provide a theoretical guarantee that, in the training phase of dynamic noise cancellation, the weights learned from time series contaminated with noise become sufficiently close to those learned from the noise-free system.

\begin{thm}[Dynamic noise cancellation]\label{theo:noisecancel}
Consider a function space $\mathcal{F} \subset \{f \colon \mathcal{Y} \to \mathcal{Y}\}$, where $\mathcal{Y}$ is a subset of $\mathbb{R}$.
We assume that a reservoir without memory ($\rho = 0$) is a universal approximator for $\mathcal{F}$ in the $L^{\infty}$ sense.
Then, for any $f \in \mathcal{F}$, $\delta > 0$ and $\varepsilon> 0$, there exist $d > 0$ and $T > 0$ such that the output $\hat{y}$, generated using the weights obtained by linear regression from noisy data $X_{\varepsilon}=(\bm{x}_\varepsilon(1) \cdots  \bm{x}_\varepsilon(T))^\top \in \mathbb{R}^{T \times d}$, $\bm{v} = (v(1) \cdots v(T))^{\top} \in \mathbb{R}^T$ satisfies the following inequality for any $y \in \mathcal{Y}$:
\begin{eqnarray}
|f(y) - \hat{y}| \leq \delta + |{\bm x}^{\top} X_{\varepsilon}^{+}\bm{v}|\varepsilon,
\label{eq:noiselessTraining}
\end{eqnarray}
where ${\bm x} = {\bm g}(y) \in \mathbb{R}^{d} $ is the reservoir state driven by input $y$, and $X_{\varepsilon}^{+} \in \mathbb{R}^{d \times T}$ is the Moore–Penrose inverse of $X_{\varepsilon}$. 
In particular, when the noise is independent and has zero mean, i.e., $\Sigma_{t=0}^{T-1} v(t)v(t+1)=0$, the reservoir trained on noisy data universally approximates $f\in \mathcal{F}$:
\begin{eqnarray}
|f(y) - \hat{y}| < \delta. \label{eq:noCorrelation}
\end{eqnarray}
\end{thm}

The proof is provided in the Appendix \ref{sec:proof}.
Regarding the universal approximation property assumed for the reservoir, echo state networks have been shown to possess universal approximation capabilities under a variety of settings~\cite{grigoryeva2018echo, gonon2021fading, sugiura2023nonessentiality, yasumoto2025universality}. 
This theorem suggests that, when noise causes the target dynamical system to behave in a non-reproducible manner, a memoryless reservoir can act as a filter that suppresses such non-reproducible components. 
This interpretation is also consistent with the well-known observation that observational noise during training can act not only as a source of degradation but also as a form of regularization~\cite{bishop1995training}. 

If the noise has a non-zero mean, that is, if it contains a systematic bias, the reservoir learns dynamics shifted by that mean.
If the noise has temporal correlations, the reservoir can instead learn part of that correlated structure, so the noise is not expected to be removed completely. 

Even when the noise takes the more general form of Eq.~(\ref{eq:generalNoise}), if $\varepsilon$ is sufficiently small and $f$ is real analytic, the system can be approximated by the additive-noise form in Eq.~(\ref{eq:additiveNoise}) through a Taylor expansion.

%% file: section/6_spin.tex
\subsection{General noise form: Current noise in spin-torque oscillator
\label{sec:sto}
}
We next consider a more practical non-additive and non-i.i.d. noise setting of the form in Eq.~(\ref{eq:generalNoise}), using a spin-torque oscillator (STO) as an example. 
STOs are spintronic devices relevant to neuromorphic computing~\cite{grollier2020neuromorphic}, and their dynamics are known to exhibit noise-induced bifurcations under piecewise-constant current input~\cite{akashi2020input-driven}.
Since bifurcations strongly affect their computational performance~\cite{akashi2020input-driven, akashi2022coupled}, predicting and suppressing noise-induced changes is of practical importance.

We simulate the STO by the Landau--Lifshitz--Gilbert equation and use the magnetization ${\bm m} = (m_x, m_y, m_z)^{\top}\in\mathbb{R}^{3}$ as the observed state (see Appendix \ref{sec:stoSetting}). 
In this system, the current enters the dynamics non-additively.

\begin{figure*}[htbp]
\begin{center}
\includegraphics[width=\linewidth]{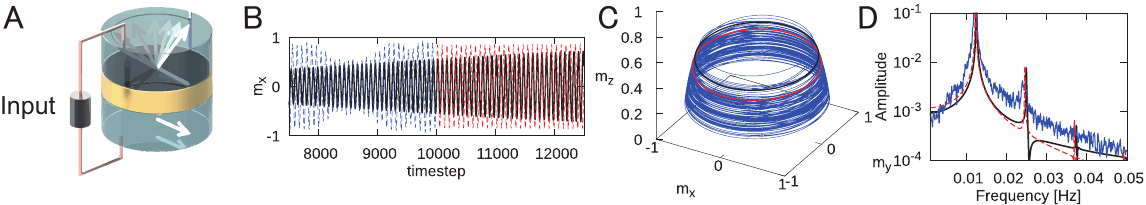}
\caption{
\label{fig:sto}
Results of dynamic noise cancellation in the STO. Black, blue, and red lines represent the target, training and reconstructed time series, respectively. 
A, Schematic illustration of the STO. 
B, Time series. The ESN is switched from open-loop to closed-loop at time step 10,000. 
C, Trajectories. 
D, Fourier spectra.
We set the parameters $(\rho, \sigma, \varepsilon_{\rm train})=(0.00, 1.00, 5.5)$.
}    
\end{center}
\end{figure*}

The results are shown in Fig.~\ref{fig:sto}. 
Although the training trajectories exhibit large fluctuations and lie in a regime with a positive maximum conditional Lyapunov exponent~\cite{akashi2020input-driven}, the reconstructed dynamics converge to a stable limit cycle absent from the training data. 
The fluctuations observed in the Fourier spectrum of the training signal are also absent from the Fourier spectrum of the reconstructed signal.
The prediction of noise-induced chaotic dynamics has been also demonstrated in the Supplementary Information.

%% file: section/7_conclusion.tex
\section{Conclusion}
We propose a simple machine learning framework for reconstructing noise-induced bifurcations from data with a single noise condition.
We experimentally demonstrate the effectiveness of the proposed framework for three well-known dynamical systems.
We provide a theoretical analysis for dynamic noise cancellation.
The proposed framework is applicable to the complex real-world model of neuromorphic spintronics.

In the proposed method, noise cancellation can be performed without requiring prior knowledge of the noise distribution or the specific manner in which the noise enters the system.
In contrast, reconstructing noise-induced bifurcations requires information about the noise distribution and how it is applied to the state variables.
These limitations are expected to be mitigated by combining the proposed approach with methods for separating noise from the underlying system dynamics (e.g., SINDy), or with parameter-aware reservoir-based approaches~\cite{kim2021teaching, kong2023reservoir} that enable prediction solely from the noise intensity. 
We present results on noise-induced chaos prediction in the STO using a parameter-aware reservoir in the Supplementary Information.

In the numerical experiments, we evaluated the prediction of bifurcations using Lyapunov exponents.
In the Lorenz-system experiments, we considered noise intensities comparable to the size of the attractor. 
Under such large noise, however, it is not clear whether the assumptions required for the existence of Lyapunov exponents under the multiplicative ergodic theorem are satisfied.
In situations where additive Gaussian noise is applied to a non-singular dynamical system, the existence of an absolutely continuous invariant measure is theoretically guaranteed~\cite{galatolo2025efficient, barrientos2025finitude}.
Thus, when evaluating reconstruction, it is necessary to identify the relevant invariant quantities under the noise conditions being considered.

The results of this study advance our understanding of complex nonlinear phenomena under noise.
In addition, the results may advance next-generation engineering applications that positively exploit complex dynamics such as neuromorphic devices.

%% file: section/A2_benchmarks.tex
\section{Benchmark systems\label{sec:benchmark}}
We consider three dynamical systems with dynamic noise as benchmark systems for evaluating the performance of the proposed framework: the logistic map, the Belousov–Zhabotinsky (BZ) map, and the Lorenz system. 
\subsection{Random logistic map}
The first system is the random logistic map~\cite{sano2020reduction}. 
This is a one-dimensional discrete-time dynamical system given by
\begin{eqnarray}
y_{\varepsilon}(t+1) = ay_{\varepsilon}(t)(1-y_{\varepsilon}(t)) + \varepsilon v(t) \; ({\rm mod}\; 1),
\end{eqnarray}
where noise $v(t)$ is drawn independent and identically distributed (i.i.d.) from the uniform distribution on $[-1,1]$.
We impose periodic boundary conditions so that the state remains in $[0,1)$. 
In the absence of noise, the logistic map exhibits the well-known period-doubling bifurcation as the parameter $a$ is varied. For $a=3.54$, the noiseless system has a period-4 orbit. 
As the dynamic noise intensity increases, the system undergoes noise-induced chaos, characterized by a positive Lyapunov exponent.

\subsection{Random BZ map}
The second system is a random BZ map~\cite{tomita1980towards}. 
The BZ reaction is a well-known chemical reaction that exhibits complex oscillatory and chaotic behavior. 
The BZ map is a one-dimensional discrete-time model of this reaction, constructed from experimental data~\cite{hudson1979experimental}, and is given by
\begin{eqnarray}
y_{\varepsilon}(t+1) = \left\{
\begin{array}{ll}
\left(a + \left(y_{\varepsilon}(t) - \frac{1}{8}\right)^{\frac{1}{3}}\right){\rm e}^{-y_{\varepsilon}(t)} + b + \varepsilon v(t) & y_{\varepsilon}(t) \leq 0.3, \\
c\left(10y_{\varepsilon}(t){\rm e}^{-10y_{\varepsilon}(t)/3}\right)^{19} + b + \varepsilon v(t) & y_{\varepsilon}(t) > 0.3,
\end{array}
\right.
\label{eq:bz}
\end{eqnarray}
where noise $v(t)$ is drawn independent and identically distributed (i.i.d.) from the uniform distribution on $[-1,1]$.
$(a, b, c)=(0.125, 0.023289, 0.121205692)$ are a parameters for which the noiseless map exhibits chaotic dynamics. 
As the noise intensity increases, the Lyapunov exponent becomes negative, leading to noise-induced order~\cite{matsumoto1983noise, galatolo2020existence}.

\subsection{Stochastic Lorenz system}
The third benchmark system is the stochastic Lorenz system~\cite{toral2001analytical}. This is a three-dimensional continuous-time dynamical system, given by
\begin{eqnarray}
\label{eq:lorenz}
\dot{y}_\varepsilon^1 &=& p(y_\varepsilon^2-y_\varepsilon^1), \\
\dot{y}_\varepsilon^2 &=& -y_\varepsilon^1y_\varepsilon^3 + ry_\varepsilon^1 - y_\varepsilon^2 + \varepsilon \dot{W}, \\
\dot{y}_\varepsilon^3 &=& y_\varepsilon^1y_\varepsilon^2 - by_\varepsilon^3,
\end{eqnarray}
where the parameters are set to $(p,b,r)=(10,8/3,20)$.
Although the Lorenz system exhibits a strange attractor for the parameter value $r=28$, it has two stable fixed points at $r=20$. 
As the noise intensity $\varepsilon$ increases, noise-induced chaos emerges through stabilization of the strange attractor.
For larger noise intensities, a further bifurcation occurs in which the Lyapunov exponent becomes negative again. 
The stochastic Lorenz system thus exhibits a two-stage noise-induced bifurcation characterized by an order--chaos--order transition.

%% file: section/A2_experimentalSetting.tex
\section{Experimental setting \label{sec:setting}}
Here, we summarize the experimental settings of the ESNs used to reconstruct each model.

\begin{table}[htbp]
    \centering
    \begin{tabular}{rcccccc}
        Target system & Node size & Activation function & $T_{\rm wash}$  & $T_{\rm train}$ & $T_{\rm eval}$ \\ \hline
        Logistic map & 100 & tanh & 100 & 9,000 & 900 \\
        BZ map & 1,000 & hard-tanh & 100 & 9,000 & 900   \\
        Lorenz system & 100 & tanh & 1000 & 90,000 & 9000 \\
        STO & 500 & tanh & 9,000 & 35,000 & 5,000 \\
    \end{tabular}
    \caption{Reservoir configurations.}
    \label{tab:esnSetting}
\end{table}
We use the ridge parameter $\beta = 0.00001$ in all experiments.

For the reservoir used to learn the BZ map, we employed the hard-tanh activation function defined as
\begin{eqnarray}
    {\bm g}(x) = \left\{
\begin{array}{ll}
-1 & (x < -1), \\
x & (-1 \leq x \leq 1), \\
1 & (x > 1),
\end{array}
\right.
\end{eqnarray}
It has been reported that hard tanh, which contains non-differentiable points, is well-suited for learning non-smooth dynamical systems~\cite{shi2024predicting}. 
Because the BZ map in Eq.~(\ref{eq:bz}) has an extremely strong nonlinearity of nineteenth order, using hard tanh as the activation function improved the learning performance of one-step prediction.

We also consider cases in which the target attractor is a periodic orbit or an equilibrium point.
In such cases, the training data contain only a limited number of distinct patterns, which can easily lead to overfitting. 
Therefore, when the training data are noise-free and the attractor is a periodic orbit or an equilibrium point, we construct the training data not from a single trajectory but by concatenating transient trajectories of length $T$ starting from $L$ different initial conditions.
No training is performed on the junctions between consecutive trajectories.
More specifically, for the training data $\left\{ \left\{ {\bm y}_{\varepsilon_{\rm train}}^l(t) \right\}_{t=1,\dots,T-1} \right\}_{l=1,\dots,L}$, the corresponding teacher data are given by $\left\{ \left\{ {\bm y}_{\varepsilon_{\rm train}}^l(t) \right\}_{t=2,\dots,T} \right\}_{l=1,\dots,L}$.
Similar training data constructed from transient trajectories have also been used in refs.~\cite{kim2021teaching, tadokoro2024trans, tokuda_prediction_2024}. 
We use $L=100$ in the training of the logistic map and the Lorenz system.

To reduce the computational cost, the largest Lyapunov exponents of both the target system and the ESN were computed using the Shimada--Nagashima method~\cite{shimada1979numerical} instead of the Jacobian calculation.

%% file: section/A1_proof.tex
\section{Proof of theorem \label{sec:proof}}
Here, we provide the proof of Theorem~\ref{theo:noisecancel}. 
First, we define that a model is a universal approximator for $\mathcal{F} \subset \{f \colon \mathcal{Y} \to \mathcal{Y}\}$ in the $L^{\infty}$ sense if, for any $f \in \mathcal{F}$ and $\delta > 0$, there exist a number of units $d > 0$ and an output weight $w^{{\rm out}}_{0} \in {\mathbb{R}^{d+1}}$ such that the following inequality holds for any $y\in \mathcal{Y}$:

\begin{eqnarray}
&&|f(y) - \hat{y}| \\
&=&|f(y) - {\bm{x}^{\top}}w^{{\rm out}}_{0}| \\
&=& |f(y) - \bm{g}(y)^{\top} w^{{\rm out}}_{0}| < \delta. \label{eq:uap}
\end{eqnarray}

\begin{proof}[Proof of Theorem \ref{theo:noisecancel}]
We first show the relationship between the output weights $w^{\rm out}_{\varepsilon}$, obtained by training on noisy data $X_{\varepsilon}=(\bm{g}(y(0)) \cdots \bm{g}(y(T-1)))^{\top}$ and $Y_{\varepsilon}=(f(y(0)) \cdots f(y(T-1)))^{\top} +\varepsilon{\bm v} = Y_0+\varepsilon{\bm v}$, and the output weights $w^{{\rm out}}_{0}$, which achieve universal approximation of the underlying noise-free dynamical system.
By linear regression, the weights $w_{{\rm out},\varepsilon}$ are derived from the data as follows:

\begin{eqnarray}
w^{\rm out}_{\varepsilon} = X_{\varepsilon}^{+}Y_{\varepsilon}.
\end{eqnarray}
We evaluate the error between the prediction and the target signals.   
\begin{eqnarray}
&& |y_0 - \hat{y}_0| \nonumber \\
&=& |y_0 - {\bm x}^{\top} w^{{\rm out}}_{\varepsilon}|  \nonumber \\
&\leq& |y_0 - {\bm x}^{\top} w^{{\rm out}}_{0}| + |{\bm x}^{\top} w^{{\rm out}}_{0} - {\bm x}^{\top} w^{{\rm out}}_{\varepsilon}|  \nonumber \\
&\leq& \delta + |{\bm x}^{\top} X_{\varepsilon}^{+} (X_{\varepsilon}w^{{\rm out}}_{0} - (Y_0 + \varepsilon {\bm v})) |\nonumber \\
&\leq& (1 + \| {\bm x}^{\top} X_{\varepsilon}^{+}\|_\infty)\delta
+ |{\bm x}^{\top} X_{\varepsilon}^{+}{\bm v}|\varepsilon. \label{eq:proofNoiseCancel}
\end{eqnarray}
Because the coefficient of the first term in Eq.~(\ref{eq:proofNoiseCancel}) is bounded, we can replace it by a constant $\delta$, which yields Eq.~(\ref{eq:noiselessTraining}).
Furthermore, the second term in Eq.~(\ref{eq:proofNoiseCancel}) can be rewritten as follows:
\begin{eqnarray}
&& \varepsilon|{\bm x}^{\top} X_{\varepsilon}^{+}{\bm v}| \nonumber \\
&=& \varepsilon \lvert {\bm x}^{\top} (X_{\varepsilon}^{\top}X_{\varepsilon})^{+}X_{\varepsilon}^{\top}{\bm v} \rvert \nonumber \\
&=& \varepsilon \left|{\bm x}^{\top} (X_{\varepsilon}^{\top}X_{\varepsilon})^{+} 
\begin{pmatrix}
\langle x_{\varepsilon, 1}(t)v(t+1) \rangle_t \\
\vdots\\
\langle x_{\varepsilon, D}(t)v(t+1) \rangle_t
\end{pmatrix}
\right| \nonumber \\
&=& \varepsilon\left|{\bm x}^{\top} (X_{\varepsilon}^{\top}X_{\varepsilon})^{+} 
\begin{pmatrix}
{\rm cov}(x_{\varepsilon, 1}(t), v(t+1)) - \langle x_{\varepsilon, 1}(t) \rangle_t \langle v(t+1) \rangle_t \\
\vdots \\
{\rm cov}(x_{\varepsilon, D}(t), v(t+1)) - \langle x_{\varepsilon, D}(t) \rangle_t \langle v(t+1) \rangle_t
\end{pmatrix}
\right| \label{eq: noiseterm}
\end{eqnarray}
Here, $\langle v(t) \rangle_t$ denotes the time average of $v(t)$ from $t=1$ to $T$, and ${\rm cov}(x(t),v(t))$ denotes the covariance between the time series $x(t)$ and $v(t)$ over $t=1,\dots,T$.
In particular, when the noise is independent, that is, when $\lim_{T\to\infty}\sum_{t=0}^{T-1} v(t)v(t+1)=0$, and has zero mean, $\lim_{T\to\infty}\langle v(t)\rangle_t=0$, this term can be made arbitrarily small by taking $T$ sufficiently large. 
Therefore, by similarly absorbing the sum of this term and the first term, which can also be made arbitrarily small by taking $d$ sufficiently large, into a constant $\delta$, Eq.~(\ref{eq:noCorrelation}) is obtained.

\end{proof}

%% file: section/A3_spinParameter.tex
\section{The LLG equation for STO simulation
\label{sec:stoSetting}
}

We simulate STO dynamics by the LLG equation in section~\ref{sec:sto}.
The LLG equation is given by the following ordinary differential equation:
\begin{eqnarray}
\label{eq:llg}
\frac{d{\bm m}}{dt}&=&-\gamma {\bm m} \times {\bf H} - \gamma H_s(t) {\bm m} \times ({\bf p}\times {\bm m}) + \alpha {\bm m} \times \frac{d{\bm m}}{dt},
\\
H_s(t)&=&\frac{\hslash \eta j(t)}{2e(1+\lambda {\bm m}\cdot{\bf p}) M V}.
\end{eqnarray}
The parameter values in the LLG equation are the same as those used in Ref.~\cite{akashi2020input-driven} and are based on an experimental study~\cite{kubota2013spin} and a theoretical study~\cite{taniguchi2017relaxation}: the saturation magnetization $M = 1448.3 {\rm emu/c.c}$, interfacial magnetic anisotropy field $H_{{\rm K}} = 18.616 {\rm kOe}$, applied field $H_{{\rm appl}} = 2.0 \; {\rm kOe}$, volume of the free layer $V = \pi \times 60^2 \times 2 {\rm nm}^{3}$, spin polarization $\eta = 0.537$, gyromagnetic ratio $\gamma = 1.764 \times 10^7 {\rm rad/(Oe s)} $, spin-transfer torque asymmetry $\lambda = 0.288$, and Gilbert damping constant $\alpha = 0.005$.
The magnetization ${\bf p}$ in the reference layer is fixed to positive $x$. 

We inject a noise signal into the current $j(t)$.
$j(t)$ is given by $j(t) = j_{\rm dc} + \varepsilon j_{\rm noise}(t)$, where $\varepsilon$ is the input intensity, $j_{\rm dc}=2.5$ mA is a constant, and the noise signal is $j_{\rm noise}(t) = u_n\;\; (n = {\rm max}\{n \in \mathbb{Z} \; | \; n < t / \tau \}) $ as a piecewise-constant signal, where ${u_n}$ is an i.i.d. random variable drawn from the uniform distribution on $[-1,1]$ and $\tau=250$ ps as the input interval.

%% file: section/S1_spin_PARC.tex
\section{Prediction of noise-induced bifurcations in the spin-torque oscillator using parameter-aware reservoir computing
\label{sec:sto}}
Here, we reconstruct noise-induced bifurcations under general-form noise using parameter-aware reservoir computing (PARC, also called reservoir computing with control input)~\cite{kim2021teaching}.
In the main text, we considered a spin-torque oscillator as an example of a dynamical system with general-form noise that exhibits noise-induced bifurcations, and performed dynamic noise-cancellation experiments. 
Because the way in which the noise enters the internal state of the system was unknown, Cases 2 and 3 of the proposed method, which address reconstruction of noise-induced bifurcations, could not be applied. 

Parameter-aware reservoir computing (PARC) explicitly feeds control parameters into the reservoir, allowing it to learn a family of target time series together with their parameter dependence. 
Moreover, PARC can interpolate and extrapolate with respect to parameter values and has been shown to reconstruct bifurcation structures~\cite{kim2021teaching}.
Here, by using a noise-dependent physical quantity as the parameter input, we make the reservoir learn the dependence on noise intensity and thereby reproduce the noise-induced bifurcations.

\begin{figure*}[htbp]
\begin{center}
\includegraphics[width=\linewidth]{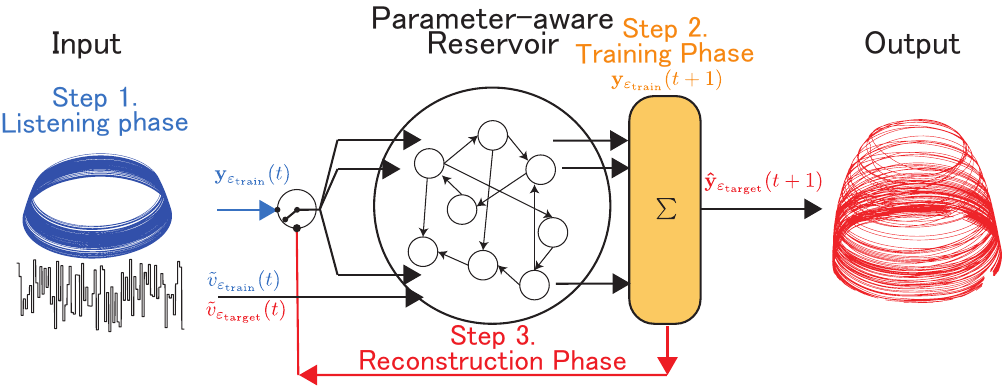}
\caption{\label{fig:parc}
Schematic illustration of the parameter-aware reservoir to reconstruct noise-induced bifurcations.
}
\end{center}
\end{figure*}

PARC adopts the architecture illustrated in Fig.~\ref{fig:parc}. The state evolution of the echo state network used here as PARC is given by
\begin{eqnarray}
{\bm x}_{\varepsilon}(t) &=& {\bm g}\left(\rho W{\bm x}_{\varepsilon}(t-1) + \sigma W^{\rm in}{\bm u}_{\varepsilon}(t) + W^{\rm param} \tilde{v}_{\varepsilon}(t) + {\bm b}\right),
\label{eq:parc}
\end{eqnarray}
where $\tilde{v}_{\varepsilon}(t)$ denotes an observed quantity of the system affected by noise.

Using PARC, we reconstruct the noise-induced bifurcations of the STO dynamics. 
As in the experiments in the main text, the STO magnetization $\bm{m}(t)=(m_x(t),m_y(t),m_z(t))^{\top}$ is provided as the reservoir input representing the dynamical state, while the current value $j(t)$ is given as a parameter input affected by noise. 
Unlike in the proposed method in the main text, the current value, rather than the noise intensity itself, is supplied as the parameter input; thus, the noise intensity is not handled explicitly.
Instead, the aim is to reproduce the dynamics and bifurcations under previously unseen current conditions.

\begin{figure*}[htbp]
\begin{center}
\includegraphics[width=0.5\linewidth]{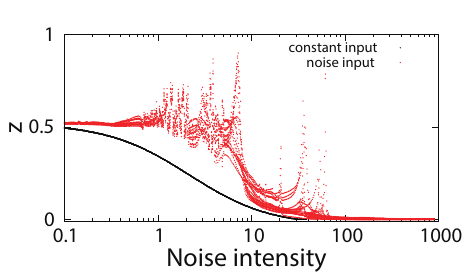}
\caption{\label{fig:stoBifurcation}
The bifurcation diagram of STO dynamics as a function of noise intensity $\varepsilon$.
}
\end{center}
\end{figure*}

We first describe the noise-induced bifurcation in the STO. 
Figure~\ref{fig:stoBifurcation} shows the bifurcation diagram of the STO dynamics as a function of the noise intensity $\varepsilon$.
For each noise intensity, the local minima of $m_z$ are plotted.
In the regime $\varepsilon < 1.0$, only a single minimum is observed, indicating limit-cycle dynamics. 
In the range $1.0 < \varepsilon < 100$, many minima appear, showing that the system has bifurcated to chaotic dynamics.
For $\varepsilon > 100$, a second bifurcation occurs, and the dynamics become non-chaotic again~\cite{akashi2020input-driven}. 
In this section, we train the reservoir on the dynamics in the regime $\varepsilon < 1.0$ and then attempt to predict the noise-induced chaos in the range $1.0 < \varepsilon < 100$.

\begin{figure*}[htbp]
\begin{center}
\includegraphics[width=\linewidth]{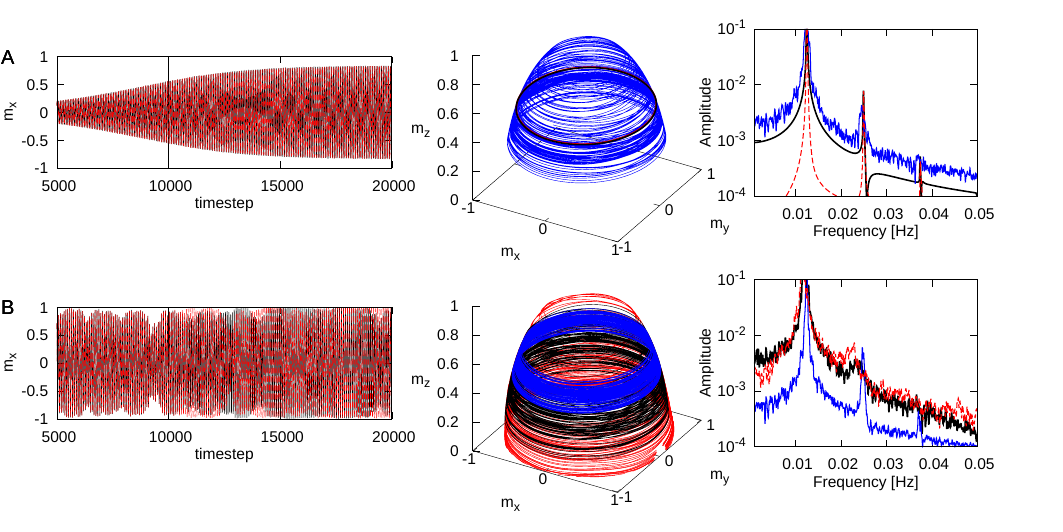}
\caption{
\label{fig:sto_parc}
Results of A. dynamic noise cancellation and B. noise-induced chaos reconstruction in the STO using PARC. 
Black, blue, and red lines denote the target, training, and reconstructed time series, respectively. 
The left panels show time series. 
The reservoir is switched from open-loop to closed-loop operation at time step 10,000. The middle panels show trajectories. The right panels show Fourier spectra. The parameters are set to $(\rho, \sigma, \varepsilon_{\rm train}, \varepsilon_{\rm target})=(0.1, 1.00, 3.85, 0.0)$ and $(0.1, 1.00, 0.891, 8.00)$ for A dynamic noise cancellation and B noise-induced chaos reconstruction, respectively.
}
\end{center}
\end{figure*}

Figure~\ref{fig:sto_parc} shows the results of dynamic noise cancellation and reconstruction of noise-induced bifurcations in the STO using PARC.
In Fig.~\ref{fig:sto_parc}A, the reservoir is trained on magnetization dynamics and current values at a noise intensity above the onset of noise-induced chaos, and we then attempt to reconstruct the dynamics under a constant current corresponding to zero noise intensity. 
The limit-cycle dynamics under constant current are reconstructed accurately, including both convergence to the attractor and the precise location of the attractor in phase space. In the dynamic noise-cancellation experiment described in the main text, the limit-cycle dynamics were reproduced, but the attractor position showed a slight discrepancy.
This suggests that, by providing the noisy current state separately from the magnetization in PARC, the reservoir is able to learn the effect of noise more accurately.

In Fig.~\ref{fig:sto_parc}B, the reservoir is trained at a relatively weak noise intensity around the onset of noise-induced chaos, and we then attempt to reconstruct the noise-induced chaotic dynamics at larger noise intensities. 
Although the precise location of the attractor is not fully reproduced, the reservoir reconstructs dynamics that spread more broadly in phase space than those seen during training, and the resulting Fourier spectrum is also close to that of the target dynamics.

